\documentclass[10pt]{article}
\usepackage[margin=1in]{geometry}

\usepackage{style/qstyle}

\title{Deciding Whether a C–Q Channel Preserves a Bit is QCMA-Complete}
\author[1]{Kiera Hutton}
\author[1]{Arthur Mehta}
\author[2]{Andrej Vukovic}
\affil[1]{ Department of Mathematics and Statistics, University of Ottawa\footnote{\texttt{khutt011@uottawa.ca, amehta2@uottawa.ca}}}
\affil[2]{ Department of Systems and Computer Engineering, Carleton University\footnote{\texttt{andrejvukovic@cunet.carleton.ca}}}
\date{\vspace{-10mm}}

\begin{document}

\pagenumbering{roman}
\setcounter{page}{1}

\maketitle

\begin{abstract}
We prove that deciding whether a classical–quantum (C–Q) channel can exactly preserve a single classical bit is QCMA-complete. This “bit-preservation” problem is a special case of orthogonality-constrained optimization tasks over C–Q channels, in which one seeks orthogonal input states whose outputs have small or large Hilbert–Schmidt overlap after passing through the channel. Both problems can be cast as biquadratic optimization with orthogonality constraints. Our main technical contribution uses tools from matrix analysis to give a complete characterization of the optimal witnesses: computational basis states for the minimum, and $\ket{+}$, $\ket{-}$ over a single basis pair for the maximum. Using this characterization, we give concise proofs of $\mathrm{QCMA}$-completeness for both problems.
\end{abstract}

\pagenumbering{arabic}
\setcounter{page}{1}

\section{Introduction}
A growing line of work has studied the complexity of encoding classical information to withstand noise arising from a single use of a quantum channel~$\Phi$~\cite{BS08,CM23,DFKR25}. In the one-shot zero-error setting, this condition can be expressed either in terms of trace distance,
\[
\tfrac{1}{2}\|\Phi(\ketbra{u})-\Phi(\ketbra{v})\|_1=1,
\]
or, equivalently, as vanishing Hilbert–Schmidt overlap,
\[
\Tr(\Phi(\ketbra{u})\,\Phi(\ketbra{v}))=0.
\]
Determining whether such a pair exists is the same as deciding whether the channel is able to exactly preserve a single classical bit.  
Our work shows this problem is $\mathrm{QCMA}$-complete for the set of classical–quantum (C–Q) channels.
\begin{theorem}[Informal]\label{thm:exact case}
Deciding whether a C–Q channel preserves a classical bit is a $\mathrm{QCMA}$-complete problem.\footnote{Throughout, we assume channels are specified succinctly by a polynomial-size quantum circuit.} 
\end{theorem}

\Cref{thm:exact case} follows from a more general analysis in which we also relax the zero-error requirement. When this requirement is relaxed, the trace-distance and overlap formulations give rise to related but distinct optimization problems.
In~\cite{DFKR25}, the problem of approximating the trace norm contraction coefficient $\eta_{\mathrm{tr}}(\Phi)=\max_{\rho,\sigma}\|\Phi(\rho)-\Phi(\sigma)\|_1/\|\rho-\sigma\|_1$ was shown to be $\mathrm{NP}$-hard. Because computing trace distance is not sample-efficient in general~\cite{HHJ+17},~\cite{DFKR25} considers channels $\Phi$ given by an explicit Kraus representation, allowing this ratio to be computed numerically. Another approach, studied in~\cite{CM23}, considers the problem of approximating the minimum (or maximum) of the Hilbert–Schmidt overlap $\Tr(\Phi(\ketbra{u})\,\Phi(\ketbra{v}))$. While the overlap of mixed states lacks the direct operational interpretation of trace distance, it is efficiently estimable via the SWAP test, in contrast to the trace distance.
In this setting, channels~$\Phi$ may act on polynomially many qubits and are specified succinctly by a polynomial-size quantum circuit~$C$. Since the overlap is efficiently estimable from $\ket{u} \otimes \ket{v}$, these optimization problems lie in $\mathrm{QMA}(2)$.

The class $\mathrm{QMA}(2)$ is one of several quantum analogues of $\mathrm{NP}$; see~\cite{Gha24} for a broader survey. Among these classes, some differ by the structure of the witness—for example, $\mathrm{QCMA}$ uses a classical string, $\mathrm{QMA}$ an arbitrary quantum state, and $\mathrm{QMA}(2)$ two unentangled quantum states. These satisfy $\mathrm{QCMA} \subseteq \mathrm{QMA} \subseteq \mathrm{QMA}(2)$, with the containments believed to be strict. While numerous $\mathrm{QMA}$-complete problems are known~\cite{Boo13}, only a few natural complete problems have been identified for $\mathrm{QMA}(2)$ and $\mathrm{QCMA}$. Examples include those in~\cite{CS12,GHMW15} for $\mathrm{QMA}(2)$ and in~\cite{GS15} for $\mathrm{QCMA}$. In~\cite{CM23}, it was shown that the overlap-optimization problems we study are $\mathrm{QMA}(2)$-complete when quantified over general quantum channels, and $\mathrm{QMA}$-complete when restricted to entanglement-breaking channels. 

In this work, we study the complexity of overlap-based optimization problems for classical-quantum (C–Q) channels. Our main result establishes $\mathrm{QCMA}$-completeness for the following problems.

\begin{theorem}\label{thm:Main}
There exist functions $c, s : \mathbb{N} \rightarrow [0,1]$ with $c - s = 1/\mathrm{poly}(n)$ such that the following promise problems are $\mathrm{QCMA}$-complete. Each takes as input a quantum circuit $C$ implementing a classical–quantum channel~$\Phi_C$.
\begin{enumerate}
    \item \textit{\small Small Overlap Problem.} Decide whether there exist orthogonal pure states $\ket{u}, \ket{v}$ such that
    \[
        \Tr\bigl(\Phi_C(\ketbra{u}{u}) \, \Phi_C(\ketbra{v}{v})\bigr) \leq 1 - c,
    \]
    or whether for all orthogonal $\ket{u}, \ket{v}$,
    \[
        \Tr\bigl(\Phi_C(\ketbra{u}{u}) \, \Phi_C(\ketbra{v}{v})\bigr) \geq 1 - s.
    \]

    \item \textit{\small Large Overlap Problem.} Decide whether there exist orthogonal pure states $\ket{u}, \ket{v}$ such that
    \[
        \Tr\bigl(\Phi_C(\ketbra{u}{u}) \, \Phi_C(\ketbra{v}{v})\bigr) \geq c,
    \]
    or whether for all orthogonal $\ket{u}, \ket{v}$,
    \[
        \Tr\bigl(\Phi_C(\ketbra{u}{u}) \, \Phi_C(\ketbra{v}{v})\bigr) \leq s.
    \]
\end{enumerate}
\end{theorem}
Here a \emph{classical–quantum} (C–Q) channel is defined as follows.

\begin{definition}[C–Q channel~\cite{HSR03}]
    A C–Q channel~$\Phi$ from an $n$-qubit space to an $m$-qubit space is specified by $2^n$ states $\sigma_1, \dots, \sigma_{2^n}$ on $m$ qubits and acts as
\begin{equation}\label{eq:C–Q Channel}
\Phi(\rho) = \sum_{i=1}^{2^n} \Tr(\ketbra{i} \rho )\,\sigma_i.
\end{equation}
\end{definition}

Given a C–Q channel $\Phi$, and input states, $\ket{u}=\sum_i \alpha_i \ket{i}$ and $\ket{v}=\sum_i\beta_i \ket{i}$, the overlap of the outputs is 
\begin{equation}\label{eq:C–Q Overlap}
     \Tr\bigl(\Phi(\ketbra{u}{u}) \, \Phi(\ketbra{v}{v})\bigr) = \sum_{i,j} |\alpha_i|^2 |\beta_j|^2 \Tr(\sigma_i \sigma_j).
\end{equation}
In more detail, we show that these problems lie in $\mathrm{QCMA}$ for any $c,s$ with inverse-polynomial gap, and that the constant-gap versions are $\mathrm{QCMA}$-hard. In particular, since the Small Overlap Problem is shown to be $\mathrm{QCMA}$-hard when $c=1$, \Cref{thm:exact case} follows as a corollary.

The more general $k$-Clique and $k$-Independent Set problems studied in~\cite{CM23} draw inspiration from quantum graph theory~\cite{DSW13,Sta16,Gan21,Mat24}, and the problems in \cref{thm:Main} correspond to the case $k=2$. When restricting to C–Q channels these problems become an instance of the well-studied class of \emph{biquadratic programs}~\cite{LNQY10}, where one optimizes a quartic form that is quadratic in two unit vectors. Even computing a constant-factor approximation for such problems is $\mathrm{NP}$-hard~\cite{LNQY10}. Our setting adds the additional requirement that the two input states be orthogonal. While quadratic optimization with orthogonality constraints is well-studied~\cite{EAS98,AMS08,WY13}, to our knowledge this is the first work to consider such constraints in biquadratic programs.

If one only considers orthogonal computational basis states $\ket{i}$ and $\ket{j}$ with $i\neq j$, the overlap reduces to $\Tr(\sigma_i\sigma_j)$, so minimizing or maximizing~\Cref{eq:C–Q Overlap} amounts to optimizing $\Tr(\sigma_i\sigma_j)$ over distinct indices. For general orthogonal states $\ket{u}=\sum_i\alpha_i\ket{i}$ and $\ket{v}=\sum_i\beta_i\ket{i}$, the overlap also includes terms of the form $|\alpha_i|^2|\beta_i|^2\Tr(\sigma_i^2)$, where $\Tr(\sigma^2_i)$ can be larger or smaller than $\Tr(\sigma_i\sigma_j)$ for all $i \neq j$.
Our main technical contribution, proved in \Cref{subsec:technical}, is that orthogonality still allows for a concise characterization of the min/max of \cref{eq:C–Q Overlap}, which are both determined by a single pair $i \neq j$. We present the following simplified version of these results below.

\begin{theorem}[Simplified version of \cref{thm:Min Overlap Bound} and \cref{thm:Max overlap}]\label{thm:Technical} 
Let $\ket{u} = \sum_i \alpha_i \ket{i}$ and $ \ket{v}= \sum_j \beta_j \ket{j}$ be orthogonal pure states in $\mathbb{C}^n$, and let $\sigma_1, \dots, \sigma_n$ be arbitrary (possibly mixed) states in $\mathbb{M}_d(\mathbb{C})$. Then:
\begin{subequations}
\begin{align}
\sum_{i,j} |\alpha_i|^2 |\beta_j|^2 \Tr(\sigma_i \sigma_j) &\geq \min \left\{ \Tr(\sigma_i \sigma_j) : i \neq j \right\},\label{eq:Technical min}  \\
\sum_{i,j} |\alpha_i|^2 |\beta_j|^2 \Tr(\sigma_i \sigma_j) &\leq \max \left\{ \frac{1}{4} \Tr\left( (\sigma_i + \sigma_j)^2 \right) : i \neq j \right\}.\label{eq:Technical max} 
\end{align}
\end{subequations}
\end{theorem}

Given \cref{thm:Technical}, containment in $\mathrm{QCMA}$ follows immediately using the SWAP test, and $\mathrm{QCMA}$-hardness is relatively straightforward, as shown in \cref{subsec:Hardness}. Thus, the core of our contribution relies more on techniques from matrix analysis than complexity theoretic arguments.

We note that our technical results are stronger than \cref{thm:Technical}, and generalize along two distinct axes: \cref{thm:Min Overlap Bound} handles non-orthogonal states $\ket{u},\ket{v}$, yielding a lower bound in terms of $\braket{u}{v}$ and $\min_{i\neq j}\Tr(\sigma_i\sigma_j)$. Meanwhile, \cref{thm:Max overlap} considers maximizing the average overlap of $k$ pairwise orthogonal states and immediately implies the $k$-Clique problem of \cite{CM23}, for C–Q channels is in $\mathrm{QCMA}$. Unfortunately, extending \cref{thm:Min Overlap Bound} in a similar fashion appears more challenging. We propose the following conjecture for future study.

\begin{conjecture}\label{conj:Main}
Let $\sigma_1, \dots, \sigma_n \in \mathbb{M}_d(\mathbb{C})$ be arbitrary quantum states. Then for any collection of $k \leq n$ of pairwise orthogonal pure states $\ket{u_1}, \dots, \ket{u_k} \in \mathbb{C}^n$, the following bound holds:
 
\begin{equation*}
\frac{1}{k(k-1)} \sum_{r \ne s} \Tr\left( \left( \sum_i |\braket{i}{u_r}|^2 \sigma_i \right) \left( \sum_j |\braket{j}{u_s}|^2 \sigma_j \right) \right)
\;\;\geq\;\;
\min_{\substack{i_1, \dots, i_k \in [n] \\ \text{distinct}}}
\frac{1}{k(k-1)} \sum_{r \ne s} \Tr\left( \sigma_{i_r} \sigma_{i_s} \right).
\end{equation*}
\end{conjecture}

\subsection{Acknowledgements} We thank Omar Fawzi for helpful comments on an earlier version of this manuscript and for discussions about \Cref{thm:exact case} and the related work \cite{DFKR25}. We also thank Connor Paddock for feedback on the manuscript and useful conversations on techniques. AM is supported by NSERC Alliance Consortia Quantum grants, reference number: ALLRP 578455 - 22 and by NSERC DG 2024-06049.

\section{Preliminaries and notation}

We assume familiarity with the basics of quantum information theory
and computational complexity theory, and refer the reader to~\cite{NC00,Wat18} for background. 
We briefly recall our notation and some specific concepts that will be used throughout the paper.

\paragraph{Notation} For a natural number $n\in\N$, write $[n]:=\{1,2,\ldots,n\}$. We write $\C^n$ for the $n$-dimensional Hilbert space with standard (computational) basis $\set*{\ket{v}}{v\in [n]}$. Pure states on $\mathbb{C}^n$ are unit vectors $\ket{u} \in \mathbb{C}^n$, while the set of mixed states refers to the set of positive-semidefinite operators $\sigma \in \mathbb{M}_n(\mathbb{C})$ with $\Tr(\sigma)=1$. The space of $n$ qubits refers to the Hilbert space $\mathbb{C}^{\lbrace 0 ,1 \rbrace^{n}}$, with computational basis $\ket{v}$ for strings $v \in \lbrace 0 , 1 \rbrace^n$. We denote the set of all bitstrings $\{0,1\}^\ast:=\bigcup_{n=0}^\infty\{0,1\}^n$. Given some $x\in\{0,1\}^\ast$, write $|x|\in\N$ for its length. We use $\ttt{poly}=\set*{f:\N\rightarrow\R}{\exists\,k,N\geq 0.\;f(n)\leq n^k\,\forall\,n\geq N}$, to denote the set of polynomially bounded functions.

\paragraph{Complexity Theory}
A \emph{language} is a subset $L \subseteq \{0,1\}^\ast$, where elements $x \in L$ are called \emph{yes-instances} and elements $x \notin L$ are \emph{no-instances}. A \emph{promise problem} is a pair of disjoint subsets $(Y, N) \subseteq \{0,1\}^\ast$, with inputs promised to lie in $Y \cup N$. A \emph{complexity class} is a collection of languages or promise problems. A function $f:\{0,1\}^\ast \rightarrow \{0,1\}^\ast$ is \emph{polynomial-time computable} if there exists a Turing machine that outputs $f(x)$ in time polynomial in $|x|$. A language $L_1$ is \emph{polynomial-time Karp reducible} to a language $L_2$ if there exists a polynomial-time computable function $f$ such that $x \in L_1$ if and only if $f(x) \in L_2$. For promise problems $(Y_1, N_1)$ and $(Y_2, N_2)$, the reduction must satisfy $f(Y_1) \subseteq Y_2$ and $f(N_1) \subseteq N_2$. A language or promise problem $L$ is \emph{$\mathsf{C}$-hard} if every problem in $\mathsf{C}$ reduces to $L$, and \emph{$\mathsf{C}$-complete} if $L \in \mathsf{C}$ as well.

\paragraph{Quantum Circuits}

The inputs for the large and small overlap problems of \cref{thm:Main} are presented as quantum circuits. We view quantum circuits as being composed of gates from a fixed universal gate set which includes Hadamard, Toffoli, and NOT gates, together with qubit preparations in the state~$\ket{0}$ and partial trace operations. We also assume a canonical form: qubit preparations first, followed by unitary gates, then taking partial traces. Given a circuit~$C$, we write $\Phi_C$ for the channel it describes, and we define the size~$|C|$ as the total number of input qubits, output qubits, and gates. Circuit descriptions are encoded as bitstrings of length polynomial in~$|C|$.

\paragraph{The SWAP test}

The SWAP test~\cite{BCWdW01} is a simple quantum subroutine for estimating the Hilbert–Schmidt overlap $\Tr(\rho \sigma)$ between two quantum states. It operates by applying a controlled-SWAP gate to $\rho \otimes \sigma$, controlled by an ancilla in the state~$\ket{+}$, and measuring the ancilla in the Hadamard basis. The probability of outcome~$0$ (accepting) is $\frac{1}{2} + \frac{1}{2} \Tr(\rho \sigma)$. In particular, when $\rho = \Phi(\ketbra{u})$ and $\sigma = \Phi(\ketbra{v})$, this provides a sample-efficient method to estimate $\Tr(\Phi(\ketbra{u})\,\Phi(\ketbra{v}))$.

\paragraph{QCMA}

There are several natural quantum analogues of the complexity class $\mathrm{NP}$. The most well-known is $\mathrm{QMA}$ (which stands for Quantum Merlin Arthur) and denotes the set of promise languages which can be decided by a polynomial time quantum verifier who is given a quantum proof. The class $\mathrm{QCMA}$, introduced by Aharonov and Naveh, is the potentially smaller class in which a polynomial time quantum verifier is given a classical proof \cite{AN02arxiv}. Below we present a definition equivalent to that given in \cite{AN02arxiv} for this class.

\begin{definition}
    Let $c,s:\N\rightarrow[0,1]$. A promise problem $Y,N\subseteq\{0,1\}^\ast$ is in $\mathrm{QCMA}_{c,s}$ if there exist polynomials $p,q:\N\rightarrow\N$ and a Turing machine $V$ with one input tape and one output tape such that
    \begin{itemize}
        \item For all $x\in\{0,1\}^\ast$, $V$ halts on input $x$ in $q(|x|)$ steps and outputs the description of a quantum circuit acting on $p(|x|)$ qubits and which outputs a single qubit.

        \item For all $x\in Y$, there exists a classical string $y\in\{0,1\}^{p(|x|)}$ such that
        \[
        \braket{1}{\Phi_{V(x)}(\ketbra{y})}{1} \geq c(|x|).
        \]

        \item For all $x\in N$ and all $y\in\{0,1\}^{p(|x|)}$, 
        \[
        \braket{1}{\Phi_{V(x)}(\ketbra{y})}{1} \leq s(|x|).
        \]
    \end{itemize}
\end{definition}
Unlike the class $\mathrm{QMA}$, $\mathrm{QCMA}$ is known to possess perfect completeness. In particular  $\mathrm{QCMA} = \mathrm{QCMA}_{1, \epsilon}$, for some negligible function $\epsilon: \mathbb{N} \rightarrow [0,1]$, \cite{JKNN12}. This result requires circuits to be written using a suitable gate set. Since the inputs for our promise problems are themselves descriptions of quantum circuits, we require this description to also use a suitable gate set such as Hadamard, Toffoli, and NOT gates.

\section{Characterization and QCMA-Completeness}
Our approach begins with matrix-analytic arguments in \Cref{subsec:technical} establishing \Cref{thm:Technical}, which characterizes the optima in the large and small overlap problems. We then use this result in \Cref{subsec:Hardness} to prove $\mathrm{QCMA}$-completeness.

\subsection{Technical Results}\label{subsec:technical}

We will need the following generalized versions of the Cauchy–Schwarz and AM--GM inequalities. \Cref{lem:S-CS} follows from a well known proof of the Cauchy–Schwarz inequality for vectors over $\mathbb{R}^n$, and \Cref{lem:AGM} is also well known. We provide a proof of each for convenience.

\begin{lemma}\label{lem:S-CS}
    Let $\sigma_1, \dots, \sigma_k \in \mathbb{M}_n(\mathbb{C})$ be a collection of matrices, and let $\alpha_1, \dots, \alpha_k, \beta_1, \dots, \beta_k \in \mathbb{C}$. Then
    \begin{align*}
     \Tr\left( \sum_i |\alpha_i|^2 \sigma_i\sum_j |\beta_j|^2 \sigma_j \right) - \Tr  \left( \sum_i |\alpha_i||\beta_i| \sigma_i\right)^2  = \frac{1}{2} \Tr\left( \sum_{i,j} \left( |\alpha_i||\beta_j| - |\alpha_j||\beta_i| \right)^2 \sigma_i \sigma_j\right).
    \end{align*}
\end{lemma}

\begin{proof}
    \begin{align*}
        &\frac{1}{2} \Tr\left( \sum_{i,j} \left( |\alpha_i||\beta_j| - |\alpha_j||\beta_i| \right)^2 \sigma_i \sigma_j\right)\\
        = &\frac{1}{2} \Tr\left( \sum_{i,j} |\alpha_i|^2 |\beta_j|^2 \sigma_i \sigma_j + \sum_{i,j} |\alpha_j|^2 |\beta_i|^2 \sigma_i \sigma_j - 2 \sum_{i,j} |\alpha_i| |\beta_j| |\alpha_j||\beta_i| \sigma_i \sigma_j\right)\\ 
        = &\frac{1}{2} \Tr\left( 2\sum_{i,j} |\alpha_i|^2 |\beta_j|^2 \sigma_i \sigma_j - 2 \sum_{i,j} |\alpha_i| |\beta_j| |\alpha_j||\beta_i| \sigma_i \sigma_j\right)\\        
        = &\Tr\left(\sum_{i,j} |\alpha_i|^2 |\beta_j|^2 \sigma_i \sigma_j\right) - \Tr\left(\sum_{i,j} |\alpha_i| |\beta_j| |\alpha_j||\beta_i| \sigma_i \sigma_j\right)\\
        = &\Tr\left( \sum_i |\alpha_i|^2 \sigma_i\sum_j |\beta_j|^2 \sigma_j \right) - \Tr  \left( \sum_i |\alpha_i||\beta_i| \sigma_i\right)^2
    \end{align*}
\end{proof}

\begin{lemma}\label{lem:AGM}
    For any self-adjoint matrices $A,B \in \mathbb{M}_n(\mathbb{C})$ we have the inequality
\[
\Tr(AB) \leq \frac{1}{2}(\Tr(A^2) + \Tr(B^2))
\]
\end{lemma}

\begin{proof}
    Simply note $\Tr(A^2) + \Tr(B^2) -2\Tr(AB) = \Tr \left( (A-B)^2 \right) \geq 0.$
\end{proof}

\begin{theorem}\label{thm:Min Overlap Bound}
    Let $\ket{u} = \sum_i \alpha_i \ket{i}$ and $ \ket{v}= \sum_j \beta_j \ket{j}$, be pure states in $\mathbb{C}^n$, and  $\sigma_1, \dots, \sigma_n $ be states in $\mathbb{M}_d(\mathbb{C})$. Then 
    \[ \sum_{i,j} |\alpha_i|^2 |\beta_j|^2 \Tr(\sigma_i \sigma_j) \geq \theta (1-\Delta^2) \]
where
\[
\theta := \min_{i \neq j} \Tr(\sigma_i \sigma_j) 
\]
and
\[\Delta := \max_{j} \frac{|\alpha_j|\,|\beta_j|\,| \braket{u}{v}|}
{\sum_{i\neq j} |\alpha_i|\,|\beta_i| + |\braket{u}{v}|} \, \text{ when } | \braket{u}{v} | \ne 0 \text{ and } 0 \text{ otherwise}.
\]

\end{theorem}
 Before proceeding to the proof we note \Cref{eq:Technical min} follows as $\Delta=0$ whenever $\ket{u}$ and $\ket{v}$ are orthogonal.

\begin{proof}
First, to simplify the notation, let $a_i = |\alpha_i|$ and $b_i = |\beta_i|$. Note that it is not possible to have $\Tr(\sigma^2_i) < \theta $ and $\Tr(\sigma^2_j) < \theta $ for distinct $i \neq j$, since otherwise by the Cauchy--Schwarz inequality we have 
$$\Tr(\sigma_i \sigma_j) < \theta.$$ 
Without loss of generality, assume that $\Tr(\sigma_i \sigma_j) \geq \theta$ for $(i,j) \neq (1,1)$. By the triangle inequality, for any two states $\ket{u}$ and $\ket{v}$ we have 
$$a_1b_1 \leq \sum_{i>1} a_i b_i + |\braket{u}{v}|.$$ 
Hence there exists $ 0 \leq r \leq 1$ such that 
$$a_1b_1 = r \left(\sum_{i>1} a_i b_i + | \braket{u}{v}|\right).$$ 
Let $s$ denote the overall value
$$ s=\sum_{i,j} a_i^2 b_j^2 \Tr(\sigma_i \sigma_j).$$
First, we will split this sum into several parts in order to isolate and eliminate the problematic $\Tr(\sigma_1^2)$ term.
\begin{align*}
s &= a_1^2 b_1^2\Tr(\sigma_1^2) + a_1^2 \sum_{i>1} b_i^2 \Tr(\sigma_1\sigma_i) + b_1^2 \sum_{i>1} a_i^2 \Tr(\sigma_1\sigma_i) + \sum_{i,j>1} a_i^2 b_j^2 \Tr(\sigma_i\sigma_j) \\
&\geq a_1^2b_1^2 \Tr(\sigma^2_1) +  r^2\sum_{i,j >1} a_i^2 b_j^2 \Tr (\sigma_i \sigma_j ) 
+ \theta \left(a_1^2 \sum_{i >1} b_i^2  + b_1^2 \sum_{i >1} a_i^2  + (1-r^2)\sum_{i,j >1}a_i^2 b_j^2 \right)\\
&= \Tr(a_1^2b_1^2\sigma^2_1) + \Tr \left( r \sum_{i >1} a_i b_i \sigma_i \right)^2 +  r^2 \left( \Tr\left(\sum_{i,j >1} a_i^2 b_j^2 \sigma_i \sigma_j\right)  - \Tr \left( \sum_{i >1} a_i b_i \sigma_i \right)^2 \right) \\
&+ \theta \left(a_1^2 \sum_{i >1} b_i^2  + b_1^2 \sum_{i >1} a_i^2  + (1-r^2)\sum_{i,j >1}a_i^2 b_j^2 \right)
\end{align*}
Next, by applying \cref{lem:S-CS} and \cref{lem:AGM} we get,
\begin{align*}
    s &\geq  2 \Tr\left( r a_1b_1 \sigma_1 \sum_{i>1} a_ib_i \sigma_i  \right) +\frac{r^2}{2} \Tr \left( \sum_{i,j>1} (a_ib_j -a_jb_i)^2 \sigma_i \sigma_j \right) \\
    &+ \theta \left(a_1^2 \sum_{i >1} b_i^2  + b_1^2 \sum_{i >1} a_i^2  + (1-r^2)\sum_{i,j >1}a_i^2 b_j^2 \right)
\end{align*}
We have now eliminated the $\Tr(\sigma_1^2)$ term, so we can factor $\theta$ out from everything:
\begin{align*}
    s &\geq \theta \Big( 2r a_1b_1\sum_{i>1}a_i b_i \Big) \\
    &+ \theta\Big( \frac{r^{2}}{2} \sum_{i,j >1}(a_ib_j -a_jb_i)^2 + a_1^2 \sum_{i >1} b_i^2  + b_1^2 \sum_{i >1} a_i^2  + (1-r^2)\sum_{i,j >1}a_i^2 b_j^2 \Big) \\
    &= \theta \Big( r a_1b_1\Big(\sum_{i>1}a_i b_i + |\braket{u}{v}|\Big) - r a_1b_1|\braket{u}{v}| + r a_1b_1\sum_{i>1}a_i b_i \Big)\\ 
    &+ \theta\Big( \frac{r^{2}}{2} \sum_{i,j >1}(a_ib_j -a_jb_i)^2 + a_1^2 \sum_{i >1} b_i^2  + b_1^2 \sum_{i >1} a_i^2  + (1-r^2)\sum_{i,j >1}a_i^2 b_j^2 \Big)
\end{align*}
Using the definition of $r$, we now substitute in an (additional) $a_1 b_1$ factor in the first term and substitute out the $a_1 b_1$ factor at the end of the line.
\begin{align*}
    s &\geq \theta \Big( a_1^2b_1^2 - r a_1b_1|\braket{u}{v}| + r^2\Big( \sum_{i>1}a_i b_i + |\braket{u}{v}|\Big)\sum_{i>1}a_i b_i \Big)\\
    &+ \theta\Big( \frac{r^{2}}{2} \sum_{i,j >1}(a_ib_j -a_jb_i)^2 + a_1^2 \sum_{i >1} b_i^2  + b_1^2 \sum_{i >1} a_i^2  + (1-r^2)\sum_{i,j >1}a_i^2 b_j^2 \Big)\\
    &= \theta \Big( a_1^2b_1^2 - r a_1b_1|\braket{u}{v}| + r^2\Big( \sum_{i>1}a_i b_i \Big)^2 + r^2 |\braket{u}{v}| \sum_{i>1}a_i b_i\\
    &+ \frac{r^{2}}{2} \sum_{i,j >1}(a_ib_j -a_jb_i)^2 + a_1^2 \sum_{i >1} b_i^2  + b_1^2 \sum_{i >1} a_i^2  + (1-r^2)\sum_{i,j >1}a_i^2 b_j^2 \Big)
\end{align*}
By \cref{lem:S-CS},
\[r^2\left( \sum_{i>1}a_i b_i \right)^2 - r^2\sum_{i,j >1}a_i^2 b_j^2 = -\frac{r^{2}}{2} \sum_{i,j >1}(a_ib_j -a_jb_i)^2,\]
which causes several terms to cancel out, allowing us to conclude
\begin{align*}
    s &\geq \theta \Big( a_1^2b_1^2 + a_1^2 \sum_{i >1} b_i^2  + b_1^2 \sum_{i >1} a_i^2  + \sum_{i,j >1}a_i^2 b_j^2 
    - r |\braket{u}{v}| \Big(a_1b_1 - r\sum_{i>1}a_i b_i)\Big) \Big)\\
    &= \theta(1 - r^2 |\braket{u}{v}|^2)\\ 
    & \geq \theta(1-\Delta^2).
\end{align*}
\end{proof}

We next proceed to the proof of \Cref{eq:Technical max} of \cref{thm:Technical}. We first record the following well known result. See for example \cite[p. 351]{Vea21}.
\begin{lemma}
\label{lem:convex}
Let $M \in M_n(\mathbb{R})$ be a real PSD matrix. Then $g(x) := x^{\top}Mx$ is a convex function on $\mathbb{R}^n$.
\end{lemma}

\begin{theorem}\label{thm:Max overlap}
    Let $\sigma_1, \dots, \sigma_n \in \mathbb{M}_d(\mathbb{C})$ be arbitrary quantum states. Then for any collection of $k \leq n$ of pairwise orthogonal pure states $\ket{u_1}, \dots, \ket{u_k} \in \mathbb{C}^n$, the following bound holds:
    \begin{equation*}
\frac{1}{k(k-1)} \sum_{r \ne s} \Tr\left( \left( \sum_i |\braket{i}{u_r}|^2 \sigma_i \right) \left( \sum_j |\braket{j}{u_s}|^2 \sigma_j \right) \right)
\;\;\leq\;\;
\max_{\substack{T \subseteq [n] \\ |T| = k}} \frac{1}{k^2} \Tr\left( \left( \sum_{i \in T} \sigma_i \right)^2 \right).
\end{equation*}
\end{theorem}

Before proceeding to the proof we note that \Cref{eq:Technical max} follows by taking $k=2$.

\begin{proof}

Let $A_r := \sum_i |\langle i|u_r \rangle|^2 \sigma_i$, and denote the left-hand side of the inequality by
$$S := \frac{1}{k(k-1)} \sum_{r \ne s} \mathrm{Tr}(A_r A_s).$$

Let the average of the $A_r$'s be
$$\bar{A} := \frac{1}{k} \sum_{r=1}^k A_r,$$

and observe that
$$\sum_{r=1}^k \mathrm{Tr}((A_r-\bar{A})^2) \ge 0.$$

Expanding this expression and simplifying gives
$$\sum_r \mathrm{Tr}(A_r^2) - 2\mathrm{Tr}\left(\left(\sum_r A_r\right)\bar{A}\right) + k\mathrm{Tr}(\bar{A}^2) \ge 0.$$

Since $\sum_r A_r = k\bar{A}$, this simplifies further to
$$\sum_r \mathrm{Tr}(A_r^2) \ge k\mathrm{Tr}(\bar{A}^2). \ \ (*)$$

Now observe that
$$\sum_{r \ne s} \mathrm{Tr}(A_r A_s) = \mathrm{Tr}\left(\left(\sum_r A_r\right)^2\right) - \sum_r \mathrm{Tr}(A_r^2) = k^2\mathrm{Tr}(\bar{A}^2) - \sum_r \mathrm{Tr}(A_r^2).$$

Using (*), we see that
$$\sum_{r \ne s} \mathrm{Tr}(A_r A_s) \le k^2 \mathrm{Tr}(\bar{A}^2) - k\mathrm{Tr}(\bar{A}^2) = k(k-1)\mathrm{Tr}(\bar{A}^2).$$

Dividing by $k(k-1)$, we obtain
$$S \le \mathrm{Tr}(\bar{A}^2).$$

Write
$$\bar{A} = \frac{1}{k} \sum_r \sum_i |\langle i|u_r \rangle|^2 \sigma_i = \frac{1}{k} \sum_i \left(\sum_r |\langle i|u_r \rangle|^2 \right)\sigma_i.$$

Let $c_i := \sum_{r=1}^k |\langle i|u_r \rangle|^2$. Then $\bar{A} = \frac{1}{k} \sum_i c_i \sigma_i$, and our inequality becomes
$$S \le \mathrm{Tr}\left(\left(\frac{1}{k} \sum_i c_i \sigma_i\right)^2\right) = \frac{1}{k^2} \mathrm{Tr}\left(\left(\sum_i c_i \sigma_i\right)^2\right).$$

Let $c := (c_i)_i$. Then the vector $c$ is the diagonal of the matrix
$$P := \sum_{r=1}^k \ketbra{u_r}{ u_r}.$$

Observe that $P$ is a rank-$k$ orthogonal projection. By the Schur--Horn theorem \cite{Sch23,Horn54}, the set of all possible diagonal vectors $c$ for $P$ forms a convex polytope $\mathcal{C}_k$ whose vertices are precisely those vectors with exactly $k$ entries equal to $1$ and all remaining entries equal to $0$. The vertices correspond to choosing a subset $T \subseteq \{1,\dots,n\}$ of size $k$ and forming the vector $v_T := \sum_{i \in T} e_i$ where the $e_i$'s are the real standard basis vectors. 

Define the function \[g(c) := \mathrm{Tr}\left(\left(\sum_i c_i \sigma_i\right)^2\right)= \sum_{i,j} c_ic_j\mathrm{Tr}\left(\sum_{i,j}\sigma_i \sigma_j\right).\] We wish to maximize $g$ over $\mathcal{C}_k$. Define $M_{ij} := \mathrm{Tr}(\sigma_i \sigma_j)$ and $M := (M_{ij})_{i,j}$. Then $M$ is positive semidefinite and $g(c) = c^{\top}Mc$ is a quadratic form defined by a positive semidefinite matrix, hence is convex by Lemma \ref{lem:convex}.

By Bauer's maximum principle \cite{Bau58}, the maximum of $g$ must occur at a vertex of $\mathcal{C}_k$. Observe that $g(v_T) = \mathrm{Tr}\left(\left(\sum_{i \in T} \sigma_i\right)^2\right)$. The maximum is therefore $\max_{|T|=k} \mathrm{Tr}\left(\left(\sum_{i \in T} \sigma_i\right)^2\right)$ and can be attained. Overall, we obtain
$$S \le \frac{1}{k^2} g(c) \le \frac{1}{k^2} \max_{c \in \mathcal{C}_k} g(c) = \frac{1}{k^2} \max_{\substack{T \subseteq [n] \\ |T| = k}} \mathrm{Tr}\left(\left(\sum_{i \in T} \sigma_i\right)^2\right).$$
\end{proof}

\subsection{QCMA Completeness}\label{subsec:Hardness}

Now that we have obtained \cref{thm:Technical}, establishing $\mathrm{QCMA}$ completeness is relatively straightforward. For convenience, we restate the small and large overlap problems formally as promise problems.

\begin{definition}[Small and Large Overlap Problems]\label{defn:Overlap problems}
Let $\mathcal{C}_{\mathrm{C\text{-}Q}}$ denote the set of quantum circuits $C$ implementing classical–quantum channels $\Phi_C$, and let $c,s:\mathbb{N}\to[0,1]$.
\begin{enumerate}
    \item \textbf{Small Overlap:} The promise problem $\mathrm{SO}_{c,s}=(Y,N)$ where
    \[
    Y=\bigl\{\,C \in \mathcal{C}_{\mathrm{C\text{-}Q}}~|~\exists~\text{orthogonal }\ket{u},\ket{v}:
    \Tr(\Phi_C(\ketbra{u})\,\Phi_C(\ketbra{v})) \le 1-c\,\bigr\},
    \]
    and
    \[
    N=\bigl\{\,C \in \mathcal{C}_{\mathrm{C\text{-}Q}}~|~\forall~\text{orthogonal }\ket{u},\ket{v}:
    \Tr(\Phi_C(\ketbra{u})\,\Phi_C(\ketbra{v})) \ge 1-s\,\bigr\}.
    \]
    
    \item \textbf{Large Overlap:} The promise problem $\mathrm{LO}_{c,s}=(Y,N)$ where
    \[
    Y=\bigl\{\,C \in \mathcal{C}_{\mathrm{C\text{-}Q}}~|~\exists~\text{orthogonal }\ket{u},\ket{v}:
    \Tr(\Phi_C(\ketbra{u})\,\Phi_C(\ketbra{v})) \ge c\,\bigr\},
    \]
    and
    \[
    N=\bigl\{\,C \in \mathcal{C}_{\mathrm{C\text{-}Q}}~|~\forall~\text{orthogonal }\ket{u},\ket{v}:
    \Tr(\Phi_C(\ketbra{u})\,\Phi_C(\ketbra{v})) \le s\,\bigr\}.
    \]
\end{enumerate}
\end{definition}
We remark that if the input states are allowed to be mixed the resulting problems are equivalent to \cref{defn:Overlap problems}. We begin with the containment direction for both the small and large overlap problems, which both follow using the SWAP test, together with \Cref{thm:Technical}.

\begin{theorem}\label{thm:SO-containment}
 Let $c,s: \mathbb{N} \rightarrow [0,1]$.
The promise problem $\mathrm{SO}_{c,s}=(Y,N)$ is in $\mathrm{QCMA}_{c',s'}$ with $c'=\frac{c}{2} $ and $s' =\frac{s}{2}$.   
\end{theorem}

\begin{proof}
Given $C \in \mathcal{C}_{\mathrm{C\text{-}Q}}$, let $\Phi_{C}: \mathbb{M}_n(\mathbb{C}) \rightarrow \mathbb{M}_d(\mathbb{C})$ be the corresponding classical-quantum channel, and let $\sigma_1, \dots, \sigma_n \in \mathbb{M}_d(\mathbb{C})$ where $\Phi(\ketbra{i}) =\sigma_i$ be the states which define its action.
The verifier expects a proof of the form $\ket{i,j} \in \mathbb{C}^n \otimes \mathbb{C}^n$, with $i \neq j$, and applies the swap test on $\sigma_i \otimes \sigma_j$. She accepts if and only if the swap test fails. If $C \in Y$ then by \Cref{eq:Technical min} there exists $i \neq j$ such that $\Tr(\sigma_i \sigma_j) \leq 1-c$. Thus the probability the verifier will accept is $\frac{1}{2} - \frac{1}{2} \Tr(\sigma_i \sigma_j) \geq \frac{c}{2}$. Conversely, if $C \in N$ then for all $i \neq j$ we have $\Tr(\sigma_i \sigma_j ) \geq 1-s$ and the verifier accepts with probability at most $\frac{s}{2}$.
\end{proof}

A similar approach applies to the large overlap problem, giving the following result.

\begin{theorem}
    Let $c,s: \mathbb{N} \rightarrow [0,1]$.
The promise problem $\mathrm{LO}_{c,s}=(Y,N)$ is in $\mathrm{QCMA}_{c',s'}$ with $c'=\frac{1+c}{2} $ and $s' =\frac{1+s}{2}$.
\end{theorem}

\begin{proof}
Given $C \in \mathcal{C}_{\mathrm{C\text{-}Q}}$, let $\Phi_{C}: \mathbb{M}_n(\mathbb{C}) \rightarrow \mathbb{M}_d(\mathbb{C})$ be the corresponding classical-quantum channel, and let $\sigma_1, \dots, \sigma_n \in \mathbb{M}_d(\mathbb{C})$ where $\Phi(\ketbra{i}) =\sigma_i$ be the states that define its action.
The verifier expects a proof of the form $\ket{i,j} \in \mathbb{C}^n \otimes \mathbb{C}^n$, with $i \neq j$.  The verifier then applies the channel $\Phi_C$ to both $\ket{u}=\frac{1}{\sqrt{2}} ( \ket{i} + \ket{j})$ and $\ket{v}=\frac{1}{\sqrt{2}} ( \ket{i} - \ket{j})$ and runs the SWAP test on the resulting states. She accepts if and only if the SWAP test accepts. By \cref{eq:Technical max} the resulting completeness/soundness of $c',s'$ follows.
    \end{proof}

We remark that \Cref{thm:Max overlap} can similarly be used to show that the $k$-Clique problem of \cite{CM23} for C-Q channels is in $\mathrm{QCMA}$. We next turn to establishing hardness for the SO and LO problems.

\begin{theorem}\label{thm:SO-Hardness}
   There exist constants $c,s$ such that promise problem $\mathrm{SO}_{c,s}$ is QCMA-hard.
\end{theorem}
\begin{proof}
    
We will show this directly from the definition of QCMA. To start we will be using the perfect completeness property which allows us to assume $\mathrm{QCMA}= \mathrm{QCMA}_{1, \epsilon}$, for some negligible function $\epsilon$ \cite{JKNN12}. Suppose that $L=(Y,N)$ is a language in $\mathrm{QCMA}_{1, \epsilon}$. For each instance $x$ of $L$ we let $\Phi_{V(x)}$ denote the corresponding verification channel, and given classical witness $y$ we let $p_y$ denote the probability that $\Phi_{V(x)}$ accepts $y$, $ p_y=  \braket{1}{\Phi_{V(x)}(\ketbra{y})}{1}$.

Given instance $x$ we let $\Phi_x$ be the C–Q channel determined by the following procedure: first measure the input in the standard basis. Run the verification circuit $\Phi_{V(x)}$ on the resulting basis vector $y$, and if it accepts, return $\ketbra{y_1, \neg y_1}$ where $y_1$ is the first bit of $y$. On the other hand, if the verification circuit fails, return $\ketbra{0,0}$. The description of a circuit which implements $\Phi_x$ can be computed in polynomial time from $x$ and the action of the channel is given by
\[
\Phi_x(\rho) = \sum_y \Tr( \ketbra{y} \rho) \Big(p_y\ketbra{y_1, \neg y_1} + (1-p_y)\ketbra{0,0}\Big).
\]
For $z \neq y$, the overlap of $\ketbra{y}$ and $\ketbra{z}$ after passing through the channel is
\begin{equation*}
\Tr\left(\Phi_x(\ketbra{y})\Phi_x(\ketbra{z}) \right) = \begin{cases}
        p_y p_z + (1-p_y)(1-p_z) & \text{ if } y_1 =z_1\\ (1-p_y)(1-p_z) & \text{ if } y_1 \neq z_1.
    \end{cases}
\end{equation*}
If $x \in N$, then by \cref{eq:Technical min} we have, for every pair of orthogonal states $\ket{u},\ket{v}$,
\[
\Tr\bigl(\Phi_x(\ketbra{u})\,\Phi_x(\ketbra{v})\bigr)
\;\geq\;(1-\epsilon)^2,
\]
which is greater than $\frac{1}{2}$ for large enough $n$.

On the other hand, if $x \in Y$ then there exists $y$ such that $p_y=1$. Taking $z$ satisfying $y_1 \neq z_1$ gives \[\Tr\left(\Phi_x(\ketbra{y})\Phi_x(\ketbra{z}) \right) =0.\]
\end{proof}

The proof of \Cref{thm:SO-Hardness} yields hardness for the case $c=1$, and together with \Cref{thm:SO-containment} establishes \Cref{thm:exact case}.  
A similar reduction shows hardness for the Large Overlap Problem.

\begin{theorem}
   There exist constants $c,s$ such that the promise problem $\mathrm{LO}_{c,s}$ is QCMA-hard.
\end{theorem}
\begin{proof}
    We will use a similar approach as above. Suppose $L=(Y,N)$ is a language in $\mathrm{QCMA}_{1,\epsilon}$ with corresponding verification channel $\Phi_{V(x)}$. Let $\mu$ denote the $1$ qubit maximally mixed state $\mu = \frac{1}{2} I_2$.
    
    Given instance $x$ we let $\Phi_x$ be the classical–quantum channel determined by the following procedure: first measure the input in the standard basis and run the verification $\Phi_{V(x)}$ on the resulting basis vector $y$. If it accepts, return $\ketbra{0}$, and if it rejects, return $\mu$. The description of a circuit which implements $\Phi_x$ can be computed in polynomial time from $x$ and the action of the channel is given by
    \[
    \Phi_x(\rho) = \sum_y \Tr(\ketbra{y}\rho) \Big(p_y\ketbra{0}  + (1-p_y)\mu\Big).
    \]
For $z \neq y$, 
\begin{equation}\label{eq:H2}
\frac{1}{4}\Tr\left( (\Phi_x(\ketbra{y})+ \Phi_x(\ketbra{z}))^2 \right) = \frac{1}{2} + \frac{(p_y + p_z)^2}{8}.
\end{equation}
If $x \in N$ then for all $y \neq z$ the right hand side of \cref{eq:H2} is bounded above by $\frac{1}{2}(1+\epsilon^2)$ and by \cref{eq:Technical max} we have, for every pair of orthogonal states $\ket{u},\ket{v}$,
\[
\Tr\bigl(\Phi_x(\ketbra{u})\,\Phi_x(\ketbra{v})\bigr)
\;\leq\;\frac{9}{16},
\]
for sufficiently large $n$. On the other hand if $x \in Y$ then let us take $y$ to satisfy $p_y=1$ and take any $z \neq y$. Then $\ket{u}= \frac{1}{\sqrt{2}}( \ket{y} + \ket{z})$ and $\ket{v} =\frac{1}{\sqrt{2}}( \ket{y} - \ket{z})$ are two orthogonal unit vectors and the right-hand side of \cref{eq:H2} is greater than or equal to $\frac{5}{8}$.

\end{proof}

\bibliographystyle{bibtex/bst/alphaarxiv.bst}
\bibliography{bibtex/qcma}

\end{document}